\documentclass[11pt]{article}
\usepackage{amsfonts}
\usepackage{amssymb}
\usepackage{amsthm}
\usepackage{newlfont}
\usepackage{amsmath,amssymb,amstext,amsthm}
\usepackage[dvips]{graphics}
\usepackage{graphicx}
 \numberwithin{equation}{section}
\newtheorem{theorem}{Theorem}
\numberwithin{theorem}{section}

\numberwithin{lemma}{section}
\newtheorem{proposition}{\bf Proposition}

\newtheorem{rem}{Remark}
\numberwithin{corollary}{section}

\numberwithin{example}{section}

\begin{document}
\title{A more flexible counterpart of a Huang-Kotz's
 copula-type}
\author{H. M. Barakat${}^*,~$ M. A. Alawady\footnote{Department of Mathematics, Faculty of Science, Zagazig University, Zagazig 44519, Egypt. }, I. A. Husseiny${}^*$ and M. A. Abd Elgawad\footnote{Department of Mathematics, Faculty of Science, Benha University, Benha 13518, Benha, Egypt. The corresponding author: M. A. Abd Elgawad, E-mail: mohamed.abdelgwad@fsc.bu.edu.eg; mohamed\_salem240@yahoo.com}}
\date{}
\maketitle
\begin{abstract}
We propose a more flexible symmetric counterpart of the Huang-Kotz's  copula of the 1st type.
Both  the counterpart and Huang-Kotz's copula of the 1st type provide the same improvement
of the correlation level. Moreover, the proposed copula includes special cases of many other extensions of the Farlie-Gumbel-Morgenstern (FGM) copula.
\end{abstract}
\section{~Introduction}
A bivariate copula is a bivariate distribution function defined on the unit square $I^{2}=[0,1]^{2}, $ which  has uniform marginals on $[0, 1].$ Copula is mainly used to describe the dependence between random variables (RVs). Any bivariate function, $C(u,v),$ is a copula if and only if (cf. \cite{Nelsen06})
\begin{eqnarray}\label{eq(1.1)}
&&  C(u, 0)=C(0, v)=0,~ C(u, 1)=u,~C(u,v)\leq 1,~ C(1, v)=v,~\forall~ (u,v) \in I^{2},\\
&\mbox{and}&C(u_{2}, v_{2})-C(u_{2}, v_{1})-C(u_{1}, v_{2})+C(u_{1}, v_{1}) \geq 0,\label{eq(1.2)}\\
&&\forall~ (u_{i},v_{i}) \in I^{2},i=1,2,~ u_{1} \leq u_{2},~v_{1} \leq v_{2}.\nonumber
\end{eqnarray}
The  twice differentiable 2-increasing property (\ref{eq(1.2)}) can be replaced by the
condition
\begin{eqnarray}\label{eq(1.3)}
c(u, v)\geq 0,
\end{eqnarray}
where $c(u,v)$ is called copula density.

The FGM copula is defined by
\begin{eqnarray}\label{eq(1.4)}
C_{FGM}(u, v;a)=u v\{1+a(1-u)(1-v)\}, ~ a \in [-1,1],\,(u,v) \in I^{2}.
\end{eqnarray}
While the  copula (\ref{eq(1.4)}) is a flexible family and valuable in many applications, a well-known limitation of this copula  is the low dependence level it permits between RVs, Spearman's Rho $\rho \in[-0.333,0.333]$ and Kendall's Tau $\tau \in[-0.222,0.222].$
Several extensions to the FGM copula have been introduced in the literature, where most of these extensions aimed to improve the correlation level. Some examples of these important extensions are \cite{BairamovK02,BekrizadehJ17,BekrizadehEtal12,Durante06}. Recently, most of these extensions were studied in different important aspects by Abd Elgawad et al. \cite{Abd ElgawadEtal20b,Abd ElgawadEtal21a,Abd ElgawadEtal21b} and Alawady et al. \cite{AlawadyEtal21a,AlawadyEtal21b}.

Huang and Kotz \cite{HuangK99} proposed  two analogous extensions, the 1st HK-FGM  and 2nd HK-FGM types which are defined respectively by
\begin{eqnarray}\label{eq(1.5)}
C_{HK-FGM1}(u, v;a,b)&=&u v\left[1+a\left(1-u^{b}\right)\left(1-v^{b}\right)\right],~b>0, a \in \left[-\frac{1}{b^{2}},\frac{1}{b}\right],\\\label{eq(1.6)}
C_{HK-FGM2}(u, v;a,b)&=&u v\left[1+a(1-u)^{b}(1-v)^{b}\right], ~ b>1,  a  \in\left[-1,\left(\frac{b+1}{b-1}\right)^{b-1}\right],
\end{eqnarray}
with $\rho \in[-0.333,0.375]$ and $\rho \in[-0.333,0.391],$ respectively.  The difference between the maximal positive correlations of the two types (\ref{eq(1.5)}) and (\ref{eq(1.6)}) is  marginal. For more details about the Huang-Kotz copulas, see the recent works  \cite{Abd ElgawadEtal20a,Barakat19a}.

Ebaid et al. \cite{Ebaid20}  proposed the following new extended FGM copula
\begin{eqnarray}\label{eq(1.7)}
C(u,v)=uv\left[1+a(1-u)(1-v)(1-bu)(1-bv)\right],~ b  \in [0,1).
\end{eqnarray}
Ebaid et al. \cite{Ebaid20}  claimed that, the admissible range of the parameters is
\begin{eqnarray}\label{eq(1.8)}
 a \in \left[-1,\frac{1}{1-b}\right]~~\mbox{and}~~ b  \in [0,1).
\end{eqnarray}
Moreover, $\rho\in[-0.333, 1]$ and $\tau\in[-0.222, 1].$ Barakat et al. \cite{BarakatEtal21}  proved that the admissible  range (\ref{eq(1.8)}) and  the claim about the correlations are   wrong. Moreover, Barakat et al. \cite{BarakatEtal21} showed that the copula (\ref{eq(1.7)}), whenever $b  \in [0,1),$  does not increase the maximum positive correlations for the FGM copula (\ref{eq(1.4)}).

In this paper, we modify the copula (\ref{eq(1.7)}) by elongating  the range of the parameter $b$ to $[-\infty,\infty],$ where $\pm\infty$ is interpreted as $\lim\limits_{b\to\pm\infty}.$ The  modified copula  would be considered as a counterpart of the  Huang and Kotz's  copula of the 1st type defined by (\ref{eq(1.5)}), in the sense that both have two shape parameters and provide the same improvement  of  the positive correlation between the dependent variables compared with the FGM copula. However, the modified copula has  an evident preference than the copula (\ref{eq(1.5)}) because   the  new copula has a simpler functional form than the copula (\ref{eq(1.5)}), for being that in order to get it, as an extension of the classical FGM copula (\ref{eq(1.3)}), we used an extra shape parameter as a multiplicative factor, while to get the copula (\ref{eq(1.5)}),  the extra shape parameter is used  as an exponent. Besides this evident motivation,  the modified copula includes some special cases of other extensions of the FGM copula.

Moreover, Ebaid et al. \cite{Ebaid20} used the copula (1.7), with a wrong admissible range (cf. \cite{BarakatEtal21}),  to estimate the reliability in a dependent stress-strength model with an application to the Egyptian finance system. There is no doubt that this important application can not be benefited from as long as the admissible range of the  copula (\ref{eq(1.7)}) has not been determined or  was specified in a wrong way. Therefore, the result of this paper enables us to use this application.
\section{The main result}
The suggested modified copula of (\ref{eq(1.7)}) is defined by
\begin{eqnarray}\label{eq(2.1)}
C(u,v;a,b)=uv\left\{1+a(1-u)(1-v)(1+bu)(1+bv)\right\},~~ (u,v)~ \in I^{2},~b\in[-\infty,\infty].
\end{eqnarray}
Clearly, the  Spearman's Rho and Kendall's Tau of the copula (\ref{eq(2.1)}) can be determined by reversing the sign of $b$ in (17) and (18) in  Ebaid et al. \cite{Ebaid20}  as $\rho=\frac{a(2+b)^{2}}{12}$ and $\tau=\frac{a(2+b)^{2}}{18},$ respectively.

In the following theorem, we determine a subset  $\Omega$ of the admissible range of the copula (\ref{eq(2.1)}), on which the copula provides an improvement of the correlation coefficients $\rho$ and $\tau.$
\begin{theorem}\label{th2.1}
The set $\Omega$ is given by $\Omega=\Omega^+\cup\Omega^-,$ where
\begin{eqnarray}\label{eq(2.2)}
\Omega^+&=&\left\{(a,b): 0\leq b\leq 1,~ -\frac{1}{(1+b)^2}\leq a \leq\frac{1}{(1+b)};~\mbox{or}~ b>1, -\frac{1}{(1+b)^2}\leq a \leq \frac{1}{(1+b)^2}\right\},\nonumber\\
\Omega^-&=&\left\{(a,b): -2\leq b\leq 0,~  -1\leq a \leq0;~\mbox{or}~ b<-2, -\frac{1}{(1+b)^2}\leq a \leq \frac{1}{(1+b)^2}\right\}.
\end{eqnarray}
\end{theorem}
\begin{proof}
The admissible range $\Omega$ will be determined relying on the conditions  (\ref{eq(1.1)}) and (\ref{eq(1.3)}).
The corresponding copula density of  (\ref{eq(2.1)}) can be written in the form
\begin{eqnarray}\label{eq(2.3)}
c(u,v;a,b)=1+a f(u,b)f(v,b),
\end{eqnarray}
where $f(u,b)=3bu^{2}+2u(1-b)-1.$  Clearly, the condition (\ref{eq(1.1)}) is satisfied, when  $0<b<1.$ Furthermore, when  $0<b<1,$ in  view of  (\ref{eq(2.3)}), we get $f(0, b)=-1$ and $f(1, b)=3 b+2(1-b)-1=b+1>0.$
On the other hand, $f^{\prime}(u,b)=\frac{d\,f(u,b)}{d\,u}=6 b u+2(1-b)>0.$ Thus, $f(u,b)$
is strictly increasing. Moreover, $\min\limits_{0\leq u\leq1} f(u, b)=-1$ and $\max\limits_{0\leq u\leq1} f(u, b)=1+b>0.$ Therefore, in order that, the condition   (\ref{eq(1.3)}) is satisfied, i.e. $c(u,v;a,b)\geq0,$  we  must have
$$1+a~\min f(u,b).\min f(v,b)\geq0 \Longrightarrow a\geq-1, $$
$$1+a~\min f(u,b).\max f(v,b) \geq0 \Longrightarrow a\leq \frac{1}{1+b}, $$
$$1+a~\max f(u,b).\max f(v,b)\geq0 \Longrightarrow a\geq \frac{-1}{(1+b)^{2}}. $$
The above restrictions on $a$ imply that $ a \in\left[\frac{-1}{(1+b)^{2}},\frac{1}{1+b}\right].$ Bearing in mind that, $C(u,v;a,0)=C_{FGM}(u, v;a),-1\leq a\leq 1$ and $C(u,v;a,1)=C_{HK-FGM1}(u,v;a,2),-\frac{1}{4}\leq a\leq \frac{1}{2},$ we get $(a,b)\in \Omega^+,$ for $0\leq b\leq 1,$ and
$ a \in\left[\frac{-1}{(1+b)^{2}},\frac{1}{1+b}\right].$ Now, consider the case $b>1.$ Clearly, while for this case the condition (\ref{eq(1.3)}) is still satisfied, but the condition (\ref{eq(1.1)}) is not,  especially for large values of $b$ unless we alter the upper bound of $a$ to be $\frac{1}{(1+b)^{2}},$  note that in this case $\lim\limits_{b\to\infty}C(u,v;\frac{\alpha}{(1+b)^{2}},b)=uv[1+\alpha uv(1-u)(1-v)]:=C_{IFGM}(u, v;0,\alpha),$ where $-1\leq\alpha\leq1,$ and
$C_{IFGM}(u, v;a,b)$ is a single iterated FGM copula,
see Proposition 1. Thus, we proved that $\Omega^+\subset\Omega.$ We turn now to the case $b<0.$ When $-2\leq b\leq 0,$ by proceeding as we have done in the case $b> 0,$ we can check that the condition   (\ref{eq(1.3)}) is satisfied if
$a \in \left[-1,\frac{1}{1+b}\right],$ but the condition $a\leq\frac{1}{1+b}$ makes $C(u,v;a,b)>1,$ for some values of $b$ (when the  value of $b$ is  equal -1, or close to it). In order to fix  the upper bound of $a,$ we invoke the general relation
$C(u,v;a,b)\leq 1+a\leq 1,$
which enables us to alter the upper bound  $\frac{1}{1+b}$ by $0.$ Thus, $(a,b)\in\Omega^-,$ if $-2\leq b\leq0,$ and $-1\leq a\leq 0.$ Finally, when $b<-2,$ we can show that the condition (\ref{eq(1.3)}) is satisfied if
$a \in \left[\frac{-1}{(1+b)^{2}},\frac{1}{1+b}\right].$ On the other hand, in order that the condition (\ref{eq(1.1)}) is satisfied (especially for large $|b|$), we can alter the admissible range of  $a$ to  be $\left[\frac{-1}{(1+b)^{2}},\frac{1}{(1+b)^2}\right].$ This proves that $\Omega^-\subset\Omega,$
as required to complete the proof.
\end{proof}
\begin{rem} We can easily check that the improvement of the correlation attains only on $\Omega^+.$ Namely, at $b=1, a=\frac{1}{2},$ we get $\rho=0.375.$ On the other hand, no improvement can be gained on $\Omega^-,$ which coincides with the remark of Barakat et al. \cite{BarakatEtal21} that the copula (\ref{eq(1.7)}) does not improve the correlation.
\end{rem}
\begin{proposition}
It is easily to check the validity of the following relations and properties:
\begin{enumerate}
\item $C(u,v;a,0)=C_{FGM}(u, v;a),-1\leq a\leq 1,$
\item $C(u,v;a,1)=C_{HK-FGM1}(u, v;a,2),-\frac{1}{4}\leq a\leq \frac{1}{2},$
\item $C(u,v;a,-1)=C_{HK-FGM2}(u, v;a,2),-1\leq a\leq 3,$
\item $\lim\limits_{b\to\pm\infty}C(u,v;\frac{\alpha}{(1+b)^2},b)=C_{IFGM}(u, v;0,\alpha),$ where $ -1\leq \alpha\leq1$ and
\begin{eqnarray*}
C_{IFGM}(u, v;a,b)&=&uv\left\{1+a(1-u)(1-v)+buv(1-u)(1-v)\right\},\\
&-&1\leq a\leq 1, a+b\geq-1, b\leq\frac{3-a+\sqrt{9-6a-3a^2}}{2},\qquad\qquad\qquad\quad
\end{eqnarray*}
is the single iterated FGM copula, which was introduced by Huang and  Kotz \cite{HuangKotz84} and further studied by Alawady et al. \cite{AlawadyEtal20}, Barakat and Husseiny \cite{Barakat20}, and Barakat et al. \cite{Barakat20b,BarakatEtal20}.
\end{enumerate}
\end{proposition}
 

\begin{thebibliography}{99}
\bibitem{Abd ElgawadEtal20a}  Abd Elgawad, M. A., Alawady, M. A., Barakat, H. M. and Shengwu Xiong. (2020a). Concomitants of generalized order statistics from Huang-Kotz Farlie-Gumbel-Morgenstern bivariate distribution: some information measures. Bull. Malays. Math. Sci. Soc., 43: 2627-2645.

\bibitem{Abd ElgawadEtal20b}  Abd Elgawad, M. A., Barakat, H. M. and Alawady, M. A.
(2020b). Concomitants of generalized order statistics under the generalization of Farlie-Gumbel-Morgenstern type bivariate distributions. Bull. Iran. Math. Soc., 47: 1045-1068. https://doi.org/10.1007/s41980-020-00427-0.

\bibitem{Abd ElgawadEtal21a} Abd Elgawad, M. A., Barakat, H. M. and Alawady, M. A. (2021a). Concomitants of generalized order statistics from bivariate Cambanis family: Some information measures. Bull. Iranian. Math. Soc. https://doi.org/10.1007/s41980-021-00532-8.

\bibitem{Abd ElgawadEtal21b} Abd Elgawad, M. A., Barakat, H. M., Xiong, S., and Alyami, S. A. (2021b). Information measures for generalized order statistics and their concomitants under general framework from Huang-Kotz FGM bivariate distribution. Entropy, 23(3), 335. https://doi.org/10.3390/e23030335.

\bibitem{AlawadyEtal20} Alawady, M. A.,  Barakat, H. M., Shengwu Xiong and Abd Elgawad, M. A. (2020). Concomitants of generalized order statistics from iterated Farlie\textendash Gumbel\textendash Morgenstern type bivariate distribution. Comm. Statist.-Theory Meth., to appear. https://doi.org/10.1080/03610926.2020.1842452.

\bibitem{AlawadyEtal21a}    Alawady, M. A., Barakat, H. M., Shengwu Xiong and Abd Elgawad, M. A. (2021a). On concomitants of dual generalized order statistics from Bairamov-Kotz-Becki Farlie-Gumbel-Morgenstern bivariate distributions. Asian-European J. Math. https://doi.org/10.1142/S1793557121501850.

\bibitem{AlawadyEtal21b} Alawady, M. A., Barakat, H. M. and Abd Elgawad, M. A. (2021b). Concomitants of generalized order statistics from bivariate Cambanis family of distributions under a general setting. Bull. Malays. Math. Sci. Soc., to appear. https://doi.org/10.1080/03610926.2020.1842452.

\bibitem{BairamovK02} Bairamov, I. and  Kotz, S.  (2002). Dependence structure and symmetry of Huang-Kotz FGM distributions and their extensions. Metrika, 56 (1): 55-72. https://doi:10.1007/s001840100158.

\bibitem{Barakat20} Barakat, H. M. and Husseiny, I. A. (2021). Some information measures in concomitants of generalized order statistics under iterated Farlie-Gumbel-Morgenstern  bivariate type. Quaestiones Math.,  44(5): 581-598. https://doi.org/10.2989/16073606.2020.1729271.

\bibitem{BarakatEtal21} Barakat, H. M. Alawady, M. A. and  Abd Elgawad, M. A. (2021a). Correction to the paper \lq\lq A new extension of the FGM copula with an application in reliability\rq\rq $~$by Rasha Ebaid, Walid Elbadawy, Essam Ahmed and Abdalla Abdelghaly. Comm. Statist.-Theory Meth., to appear. https://doi.org/10.1080/03610926.2021.1879864

\bibitem{Barakat20b} Barakat, H. M., Nigm, E. M. and  Husseiny, I. A. (2020).  Measures of information in order statistics and their concomitants for the single iterated Farlie-Gumbel-Morgenstern bivariate distribution. Math. Popul. Studies, to appear. https://doi.org/10.1080/08898480.2020.1767926.

\bibitem{Barakat19a} Barakat, H. M., Nigm, E. M. and Syam, A. H. (2019). Concomitants of ordered variables from Huang-Kotz FGM type bivariate-generalized exponential distribution. Bull. Malays. Math. Sci. Soc., 42(1): 337-353. https://doi.org/10.1007/s40840-017-0489-5.

\bibitem{BarakatEtal20} Barakat, H. M., Nigm, E. M., Alawady, M. A. and Husseiny, I. A. (2021b). Concomitants of order statistics and record values from iterated FGM type bivariate-generalized exponential distribution. REVSTAT, 19(2): 291-307.

\bibitem{BekrizadehJ17} Bekrizadeh, H. and  Jamshidi, B. (2017). A new class of bivariate copulas: Dependence measures and properties. Metron 75 (1): 31-50. DOI: 10.1007/s40300-017-0107-1.

\bibitem{BekrizadehEtal12} Bekrizadeh, H.,  Parham,  G.  A. and  Zadkarmi, M. R.  (2012). The new generalization of Farlie-Gumbel-Morgenstern copulas. App. Math. Sci., 6 (71): 3527-3533.

\bibitem{Durante06} Durante, F. (2006). A new class of symmetric bivariate copulas. J. Nonparametric Statisti., 18 (7-8): 499-510. DOI: 10.1080/10485250701262242.

\bibitem{Ebaid20} Ebaid, R.,  Elbadawy, W.,  Ahmed, E. and  Abdelghaly, A. (2020).  A new extension of the FGM copula with an application in reliability. Comm. Statist. Theory and Meth., to appear. https://doi.org/10.1080/03610926.2020.1785501.

\bibitem{HuangKotz84} Huang, J. S. and  Kotz, S. (1984). Correlation structure in iterated Farlie-Gumbel-Morgenstern distributions.
Biometrika, 71(3): 633-636. https://doi.org/10.2307/2336577.

\bibitem{HuangK99} Huang, J. S. and  Kotz, S. (1999). Modifications of the Farlie-Gumbel-Morgenstern distributions. A
tough hill to climb. Metrika 49(2): 135-145. https://doi.org/10.1007/s001840050030.

\bibitem{Nelsen06} Nelsen, R. B. (2006). An Introduction to Copulas. 2nd ed. New York: Springer.
\end{thebibliography}
\end{document}